\documentclass[a4paper,11pt]{article}
\pdfoutput=1 

\usepackage{jheppub} 
\usepackage{etoolbox}
    \makeatletter
    \patchcmd{\maketitle}{\@fpheader}{}{}{}
   \makeatother
\usepackage[T1]{fontenc} 
\usepackage{soul}
\usepackage{mathrsfs}
\usepackage{amsthm}
\usepackage{mathtools}
\usepackage{dsfont}
\usepackage{cancel}
\usepackage{enumerate}

\newcommand{\be}{\begin{equation}}
\newcommand{\ee}{\end{equation}}
\def\beqa{\begin{eqnarray}}
\def\eeqa{\end{eqnarray}}
\def\bean{\begin{eqnarray*}}
\def\eean{\end{eqnarray*}}

\newcommand{\eqn}[1]{(\ref{#1})}
\newcommand{\del}{\partial}

\newcommand{\R}{\mathbb{R}}
\newcommand{\C}{\mathbb{C}}
\newcommand{\bd}{\mathbf d}
\newtheorem{thm}{Theorem}[section]
\newtheorem{prop}{Proposition}[section]

\title{\boldmath Jacobi sigma models}

\author[a,b]{F. Bascone}
\author[a]{Franco Pezzella}
\author[a,b]{Patrizia Vitale}

\affiliation[a]{INFN - Sezione di Napoli, \\ Complesso Universitario di Monte S. Angelo Edificio 6, via Cintia, 80126 Napoli, Italy}
\affiliation[b]{Dipartimento di Fisica ``E. Pancini'', Universit\`a di Napoli Federico II, \\ Complesso Universitario di Monte S. Angelo Edificio 6, via Cintia, 80126 Napoli, Italy}

\emailAdd{francesco.bascone@na.infn.it}
\emailAdd{franco.pezzella@na.infn.it}
\emailAdd{patrizia.vitale@na.infn.it}

\abstract{We introduce a two-dimensional sigma model associated with a Jacobi manifold. The model is a generalisation of a Poisson sigma model providing a topological open string theory. In the Hamiltonian approach first class constraints are derived, which generate gauge invariance of the model under diffeomorphisms. The reduced phase space is finite-dimensional.  By introducing a metric tensor on the target, a non-topological sigma model is obtained, yielding a Polyakov action with metric and $B$-field, whose target space is a Jacobi manifold. }


\begin{document} 
\maketitle
\flushbottom
\section{Introduction}
\label{intro}

Jacobi sigma models are here introduced as a natural generalisation of Poisson sigma models. The latter, first introduced in the context of two-dimensional gravity \cite{Ikeda1994, Schaller1994}, have been widely investigated in relation with symplectic groupoids, BF theory, branes and deformation quantisation \cite{Cattaneo2001, Cattaneo2001a, Cattaneo2000, Cattaneo2001b, Bonechi2005, Ikeda2017, Falceto2010, Calvo2006, Calvo2005, Cattaneo2013}.  They were also analysed from the point of view of holography and noncommutative geometry in \cite{Vassilevich2013}. In two dimensions these are topological field theories on a Riemannian surface $(\Sigma, g)$, with target space a Poisson manifold, $(M, \Pi)$  and a first order action 
\be \label{sigmaP}
S= \int_\Sigma  \left[ \eta_i\wedge dX^i+\frac{1}{2}\Pi^{ij}\eta_i\wedge \eta_j \right]
\ee
where $X$ is a smooth map, $\eta, \; dX$ are one-forms on $\Sigma$ with values in the pull-back of the cotagent  and tangent bundle respectively, 
$ X:\Sigma \rightarrow M, \;\;\, \eta\in \Omega^1 (\Sigma, X^* (T^*M))$, 
$dX\in\Omega^1(\Sigma,  X^*(TM))$ and $\Pi$ is a Poisson structure on $M$, namely a skew-symmetric bi-vector field satisfying Jacobi identity. When the latter is invertible, it is possible to eliminate the auxiliary field $\eta$,   and obtain a    formulation with target the  tangent space $TM$. This yields a topological sigma model, the so-called A-model  \cite{Witten1988, Witten1998},  with only a $B$-field term, being $B=\Pi^{-1}$ and $dB=0$. 
Because of the properties of the Poisson bi-vector field $\Pi$, 
a number of interesting facts are proven in the literature. First, under suitable assumptions, the constrained manifold, $\mathcal{C}$,   quotiented with respect to symmetries, is a finite dimensional symplectic groupoid, generalising a well known result which holds for $M=\mathfrak{g}^*$, the dual of a given  Lie algebra, where  $\mathcal{C}$/Sym is found to be diffeomorphic to $ T^*G$ \cite{Cattaneo2001}.  Second, the path integral quantisation of the model furnishes a field-theoretical proof of Kontsevich star product quantisation of Poisson manifolds \cite{Cattaneo2001a, Cattaneo2000}. Moreover, the model is gauge  invariant under space-time diffeomorphisms and the algebra of gauge parameters closes under Koszul bracket \cite{Cattaneo2001}. Finally, if the Lagrangian in \eqn{sigmaP} is complemented with a dynamical term $\frac{1}{2} G^{ij} \eta_i\wedge \star \eta_j $, with $G$ a metric tensor on $M$, by integrating away the auxiliary field $\eta$ it is possible to retrieve the full  Polyakov string action (see for example \cite{Schupp2012}). It is also possible to twist the Poisson structure by generalising the Poisson sigma model with the introduction of a Wess-Zumino term \cite{Klimcik2002}.

A natural question for us is then, whether it is possible  to relax the condition that $\Pi$ be Poisson, namely $[\Pi,\Pi]_S=0$\footnote{$[\;,\;]_S$ is the Schouten-Nijenhuis bracket.}. An almost obvious generalisation, although not considered insofar in the literature\footnote{While being in the process of submitting  the manuscript we have been aware of a new submission on the archives \cite{Chatzistavrakidis2020} where  similar ideas  are explored. } is to consider a Jacobi structure,  $(M, \Pi, E)$, with $\Pi$ a bi-vector field and $E\in \mathfrak{X}(M)$ a vector field on $M$ such that
\be
[\Pi,\Pi]_S= 2E\wedge\Pi\;\;\;\, {\rm and} \;\;\;\,[E,\Pi]_S=0.
\ee
The goal is thus to  build and study a two-dimensional sigma model with target space a Jacobi manifold. To this, we start from the observation, proven in \cite{Lichnerowicz1978}, that   a Jacobi structure on $M$ always gives rise to a Poisson structure on $M\times \R$, say $P$, with the help of a kind of  dilation vector field. Hence, a Poisson sigma model may  be  defined on $(M\times \R, P)$ whose dynamics  may  be reduced by means of a projection to the Jacobi manifold, $M\times \R\stackrel{\pi}{\rightarrow} M$. Thus we show that the projected dynamics can be obtained directly from an action functional on the Jacobi manifold, solely in terms of its defining structures. Such a model is shown to be topological, with first class constraints and gauge invariant under diffeomorphisms. A main result of the paper is Theorem \ref{mainthm}, which proves that, similarly to Poisson sigma models, the quotient manifold $\mathcal{C}/\rm{Diff}(\Sigma) $, with  $\mathcal{C}$ the constrained manifold, is finite dimensional, though  with dimension equal to $2 \, \rm{dim} M-2$. 

The paper is organised as follows. In Section \ref{secpoissonsigma} we shortly review the Poisson sigma model, mainly following notations and conventions of \cite{Cattaneo2001}. In Section \ref{secjacobisigma} Jacobi brackets and Jacobi manifolds are introduced and a Poisson sigma model on the extended manifold $M\times \R$ is defined. An action on the Jacobi manifold is thus proposed, which reproduces the projected dynamics. The model exhibits first class constraints, which generate gauge transformations, {but also second class ones, which shall be taken into account}. On using a consistent definition of Hamiltonian vector fields for Jacobi manifolds (see for example \cite{Vaisman2002, Asorey2017}), we show that the latter can be associated with gauge transformations and verify that they close under Lie bracket, generating space-time diffeomorphisms. The model results to be topological, with a finite number of degrees of freedom, on the boundary. In section \ref{secfulljacobi} we investigate the possibility of introducing a metric term, in analogy with what is done for   Poisson sigma models, so to obtain a model which is non-topological. We manage to integrate out the auxiliary fields and obtain a Polyakov action, with metric $g$ and $B$-field determined in terms of the defining structures of the Jacobi bracket, $(\Lambda, E)$. 

In order to better understand the novelties and peculiarities of the model, we build in Section \ref{secExamples} a noteworthy  example with the group manifold of $SU(2)$ as target space. We conclude with final remarks and perspectives.

\section{Poisson sigma models}
\label{secpoissonsigma}

Let $\left(M, \Pi \right)$ be a Poisson manifold, where $\Pi \in \Gamma(\wedge^2 TM)$ is a Poisson structure on the smooth $m$-dimensional manifold $M$, and $\Sigma$  a $2$-dimensional orientable smooth manifold, eventually with boundary. The topological Poisson sigma model is defined by the fields $(X, \eta)$, with $ X:\Sigma \rightarrow M$ and $ \eta\in \Omega^1 (\Sigma,  X^*(T^*M))$ a one-form on $\Sigma$ with values in the pull-back of the cotangent bundle over $M$. {$PM$ shall indicate the configuration space of the model, namely the space of smooth maps $X: \Sigma \rightarrow M$.}
The embedding of $\Sigma$ in $M$ is thus realised  by the fields $X$, while  $\eta$ will be associated with conjugate momenta, as we shall see below. The action functional, represented by Eq. \eqn{sigmaP}, yields the following 
 equations of motion (e.o.m.):
\begin{equation}\label{eqmopoisson1}
dX^i+\Pi^{ij}(X)\eta_j=0,
\end{equation}
\begin{equation}\label{eqmopoisson2}
d\eta_i+\frac{1}{2}\partial_i \Pi^{jk}\eta_j \wedge \eta_k=0.
\end{equation}
Consistency of the e.o.m. requires that $\Pi$ satisfies $[\Pi, \Pi]_S=0$. $[\;,\;]_S$ is the Schouten-Nijenhuis bracket,
namely a skew-symmetric bilinear map $\Lambda^p(M) \times \Lambda^q(M) \to \Lambda^{p+q-1}(M)$ given by
\be
\left[A_1 \wedge \dots \wedge A_p , \,  B_1 \wedge \dots \wedge B_q\right]_S=\sum(-1)^{t+s} A_{1} \wedge \dots \widehat{A}_{s} \dots \wedge A_{p} \wedge\left[A_{s}, B_{t}\right] \wedge B_{1} \wedge \dots \widehat{B}_{t \cdots} \wedge B_{q}
\ee
where $A_1, ..., A_p, B_1, ..., B_q$ are vector fields over $M$ and $\widehat{A}$ indicates the omission of the vector field $A$. Explicitly, we have
\be
0=[ \Pi, \Pi]^{ijk}_S= \Pi^{i\ell}\del_\ell\Pi^{jk} + {\rm cycl}(ijk)
\ee
reproducing  the Jacobi identity, which holds true for a Poisson structure. 

Note that if the worldsheet $\Sigma$ has a boundary, the boundary conditions $\eta(u)v=0 \, \, \forall \, v \in T(\partial \Sigma)$, with $u \in \partial \Sigma$, are chosen.

The  sigma model action \eqn{sigmaP} contains a number of different interesting models. For example, the most natural one corresponds to the choice  $\Pi^{ij}=0$, in which case one has simply an Abelian BF theory with action $\int_{\Sigma} d^2u \, \epsilon^{\mu\nu} \eta_{\mu i} \partial_{\nu}X^i$, while an interesting nontrivial case has a linear  Poisson structure on $M$,  $\Pi^{ij}={f^{ij}}_k X^k$. The latter  leads to a non-Abelian BF theory with action $S=\int_{\Sigma} d^{2}u \, \left(\epsilon^{\mu \nu} \eta_{\mu i} \partial_{\nu} X^{i}+ \frac{1}{2} \epsilon^{\mu \nu} {f^{i j}}_k X^{k} \eta_{\mu i} \eta_{\nu j}\right)$. In fact, in this case the Jacobi identity for $\Pi$ becomes a Jacobi identity for the structure constants of a Lie algebra ${f^{ij}}_k$. Another special case is the one with non-degenerate Poisson structure, which can be inverted to a symplectic form $\omega$ (which plays the role of   $B$-field in the language of strings), leading to the so-called A-model, with action 
$S=\int \omega_{ij} dX^i\wedge dX^j$. It is also possible to show that $2$-dimensional Yang-Mills, $R^2$-gravity theories and gauged WZW models can be obtained \cite{Ikeda2017,Schallera}.

We will now focus on the Hamiltonian approach. Let us choose locally a time coordinate $u^0=t$ and denote with $u^1=u$ the space coordinate, which   can be taken to belong to a closed interval, $u \in [0,1]$, if one wants to   describe open strings. By denoting  $\beta_i = \eta_{0i}$, $\zeta_i = \eta_{1i}$ and   $\dot{X}=\partial_t X$,  $X'=\partial_u X$, the first order Lagrangian can be written as
\begin{equation}\label{lagpo}
L(X, \zeta; \beta)=\int_I du \left[-\zeta_i \dot{X}^i+\beta_i\left(X'^i+\Pi^{ij}(X)\zeta_j \right) \right],
\end{equation}
from which it is clear that $X$ and ${-\zeta}$ are canonically conjugate variables, with Poisson brackets $\{\zeta_i(u), X^j(v)\}={-{\delta_i}^j \delta(u-v)}$ and all  other brackets vanishing. Given the explicit expression of the Lagrangian, we
notice that the action is invariant under the exchange $\dot{X} \leftrightarrow X'$ and $ \beta \leftrightarrow - \zeta$.

 Since $\beta$ has no conjugate variable, it has to be understood as a Lagrange multiplier imposing the constraints
\begin{equation}\label{constraintpoissonsigma}
\quad X'^i+\Pi^{ij}(X)\zeta_j=0.
\end{equation}
Therefore,  the Hamiltonian 
\begin{equation}
H_{\beta}={-}\int_I du \,\beta_i \left[X'^i+\Pi^{ij}(X)\zeta_j \right],
\end{equation}
is a pure constraint and the space of solutions, say  $\mathcal{C}$, can be equivalently defined as the set of common zeroes of $H_{\beta}$. It is also possible to prove \cite{Cattaneo2001} that these constraints are first class, namely  they satisfy the following relations, provided that $\beta, \beta'$  vanish on the boundary:
\be
\{ H_{\beta}, H_{\beta'}\}= H_{[\beta, \beta']} 
\ee
with 
\be[\beta, \beta']= d \langle \beta, \Pi(\beta') \rangle-\iota_{\Pi(\beta)} d\beta' + \iota_{\Pi(\beta')} d\beta
\ee
being the Koszul bracket of one-forms, which closes thanks to the Jacobi identity of $\Pi$.  Here $\langle \cdot, \cdot \rangle$  denotes the natural pairing between vectors and one-forms at a point in $M$. 
Being the Hamiltonian of the model a pure constraint, the system is invariant under time-diffeomorphisms. The infinitesimal generators are  the  Hamiltonian vector fields associated with $H_{\beta}$, 
\be 
\xi_{\beta}= \{H_\beta, \cdot \}= \dot X^i \frac{\del}{\del X^i} + \dot \zeta_i \frac{\del}{\del \zeta_i}
\ee
 where $(\dot X^i, \dot  \zeta_i) $ can be read from Eqs. \eqn{eqmopoisson1}-\eqn{eqmopoisson2} as:
\begin{equation}\label{sympoisson1}
 \dot X^i=-\Pi^{ij}\beta_j,
\end{equation}
\begin{equation}\label{sympoisson2}
\dot  \zeta_i=\partial_u \beta_i-\partial_i \Pi^{jk}\zeta_j \beta_k \,. 
\end{equation}
Moreover, by indicating with $f(u) \del_u$  a generic space diffeomorphism, it is immediate to check that this is the generator of an infinitesimal symmetry for the model, it being the Hamiltonian vector field associated with $H_\beta$, for   $\beta_\ell= f(u) \zeta_\ell$.   This is a direct consequence of the invariance of the action under the exchange  $X'\leftrightarrow \dot X$ and $\beta\leftrightarrow - \zeta$. Thus, the model is invariant under space-time diffeomorphisms and  the  reduced phase space can be defined as $\mathcal{G}=\mathcal{C}/\text{Diff}(\Sigma)$. It can be proven \cite{Cattaneo2001, Cattaneo2001a} that the latter is a finite-dimensional, closed  subspace of phase space, of dimension 2dim$(M)$, with a natural groupoid structure. Under certain conditions this is a symplectic groupoid integrating the Lie algebroid associated with the Poisson manifold \cite{Levin2000}. 

Finally, the absence of an Hamiltonian implies that there is no dynamics and the model is topological in the bulk.

\section{Jacobi sigma models}
\label{secjacobisigma}

In order to formulate a consistent sigma model with target configuration space a  Jacobi manifold $M$, we  first briefly review the main definitions of Jacobi brackets and Jacobi manifold. [See for example \cite{Vaisman2002,Marle1991,Grabowski2001,Crainic2007, Kirillov1976} and refs therein. Also see \cite{LV} for a generalization in terms of  Jacobi structures on complex line bundles]. Hence we will build a Poisson sigma model on the extended Poisson manifold $M \times \R$, according to a  {\it Poissonization procedure} of the Jacobi structure. We will thus project the obtained dynamics on the underlying Jacobi manifold and finally propose a consistent model, directly defined on the Jacobi manifold, whose dynamics is proven to coincide with the projected one. 

\subsection{Jacobi brackets and Jacobi manifold}\label{Jacobiman}

Jacobi brackets are defined by means of a bi-differential operator acting on the algebra of functions  on a smooth manifold $M$, as
\begin{equation}
\{f, g\}_J=\Lambda(d f, d g)+f(E g)-g(E f),
\end{equation}
where $\Lambda$ is a bivector field and $E$ is a vector field (called Reeb vector field) on the manifold $M$,  satisfying
\begin{equation}\label{jacobilambda}
[\Lambda, \Lambda]_S=2 E \wedge \Lambda, \quad [\Lambda, E]_S=\mathscr{L}_{E}\Lambda=0.
\end{equation}
For later convenience, we also report their explicit expression in coordinates:
\begin{equation}\label{jacobiidentitygen}
{\Lambda^{pi} \partial_p \Lambda^{jk}+{\rm cycl\, perm}\{ijk\}
=E^i \Lambda^{jk} +{\rm cycl\, perm}\{ijk\}},
\end{equation}
\begin{equation}\label{liederivelambda}
E^k \partial_k \Lambda^{ij}-\Lambda^{kj} \partial_k E^i-\Lambda^{ik} \partial_k E^j=0.
\end{equation}
Jacobi brackets are skew-symmetric and satisfy Jacobi identity just like Poisson brackets, but in general a Jacobi structure does not satisfy Leibniz rule, which is instead replaced by the condition
\begin{equation}
\{f, g h\}_J=\{f, g\}_J h+g\{f, h\}_J+g h(E f).
\end{equation}
In other words, the Jacobi bracket endows the algebra of functions $\mathcal{F}(M)$ with the structure of a Lie algebra, but, unlike the Poisson bracket,  it is not a derivation of the point-wise product among functions.
Clearly,  Jacobi brackets are  a generalisation of  Poisson brackets since the latter can be obtained from the former  if the Reeb vector field is vanishing, $E=0$. 

Analogously to the Poisson framework, an Hamiltonian vector field $\xi_f$ can be associated with a function $f\in \mathcal{F}(M)$, according to the following definition (see for example \cite{Vaisman2002}):
\begin{equation}\label{Hamvec}
\xi_f=\Lambda(df, \cdot)+f E.
\end{equation}
The map $f\rightarrow \xi_f$ is homomorphism of  Lie algebras, it being  $[\xi_f, \xi_g]=\xi_{\{f,g\}_J}$, where the bracket $[\cdot, \cdot]$ is the standard Lie bracket of vector fields.

Examples of Jacobi manifolds are locally conformal symplectic manifolds and contact manifolds. The former ones are even-dimensional manifolds endowed with a two-form $\omega$ and  an open covering of charts $\{U_{i}\}$ such that locally the restriction $\omega_{|U_i}=e^{a_i}\Omega_i$, with $\Omega_i$ symplectic form on the chart $U_i$ and $a_i$ smooth functions on the local chart. Locally, they have then a Poisson structure $\{ ,\}_i$ but globally $e^{-a_i}\{ e^{a_i}f, e^{a_i} g\}$ is a Jacobi bracket. More explicitly \cite{Marle1991}, 
one can define a locally conformal symplectic manifold by a pair $(\omega, \alpha)$, with $\omega$ a two-form with rank equal to the dimension of the manifold and $\alpha$ a one-form, such that 
\begin{equation}
d\alpha=0, \quad d\omega+\alpha \wedge \omega=0.
\end{equation}
The  Jacobi structure ($\Lambda$, $E$) is thus defined as the unique bi-vector field and the unique vector field which satisfy:
\begin{equation}
\iota_E \omega=-\alpha, \quad \iota_{\Lambda(\gamma)}\omega=-\gamma \quad \forall \, \gamma \in T^*M.
\end{equation}

Contact manifolds are instead odd-dimensional manifolds which are endowed with a  contact form (or contact structure), i.e. a one-form satisfying $\vartheta \wedge (d\vartheta)^n \neq 0$ everywhere, where $2n+1$ is the dimension of the manifold. This means that a one-form $\vartheta$ is a contact structure on a odd-dimensional manifold if $\vartheta \wedge (d\vartheta)^n$ is a volume form. Contact forms are defined up to multiplication by a non-vanishing function. It is possible to endow the algebra of functions on a contact manifold with  a Lie algebra structure \cite{Asorey2017}, which reads
\begin{equation}\label{con1}
[f, g ]\vartheta \wedge(d \vartheta)^{n} : =(n-1) d f \wedge d g \wedge \vartheta \wedge(d \vartheta)^{n-1}+(f d g-g d f) \wedge(d \vartheta)^{n}.
\end{equation}
The latter is local by construction and satisfies Jacobi identity. It is possible to show that this is actually a Jacobi bracket by defining $\Lambda$ and $E$ as follows:
\begin{equation}\label{contactcond}
\begin{aligned}
{} & \iota_E \left( \vartheta \wedge(d \vartheta)^{n}\right)=(d \vartheta)^{n} \\ &
\iota_{\Lambda} \left(\vartheta \wedge(d \vartheta)^{n}\right)=n \vartheta \wedge(d \vartheta)^{n-1}.
\end{aligned}
\end{equation}
The latter trivially imply that
\begin{equation}\label{contactcond2}
\iota_E \vartheta=1, \quad \iota_{E} d\vartheta=0,
\end{equation}
as well as 
\begin{equation}\label{contactcond3}
\iota_{\Lambda} \vartheta=0, \quad  \iota_{\Lambda} d\vartheta=1.
\end{equation}
An interesting property of a contact manifold is that its Poissonization is actually a Symplectification, as will be further commented  in the next section.

Interesting examples of contact manifolds are three-dimensional semi-simple Lie groups, where one of the basis left- or right-invariant one-forms can be chosen  as a contact structure. Especially interesting to us is the group $SU(2)$, whose associated sigma models have been widely studied.  Besides being simple and fairly well behaved in many respects, $SU(2)$ is the prototypical example of a Poisson-Lie group. It has been investigated in relation with Poisson sigma models in  \cite{Bonechi2005, Calvo2003}.  Moreover, Poisson-Lie duality of the $SU(2)$ Principal Chiral model, with and without Wess-Zumino term, has been considered by the authors in  \cite{Marotta2019, Bascone2020,Bascone2020a}. Therefore, we are interested in the possibility of generalising previous results obtained in \cite{Marotta2019, Bascone2020,Bascone2020a} to Jacobi sigma models on $SU(2)$ and we will exhibit a preliminary analysis in Section \ref{secExamples}.

\subsubsection{Homogeneous Poisson structure on $M \times \mathbb{R}$ from Jacobi structure}

The starting point for the subsequent analysis is provided by  the following theorem \cite{Lichnerowicz1978}:
\begin{thm}\emph{
$J(f,g)=\Lambda(df, dg)+f(E g)-g(E f)$ defines a Jacobi structure on the manifold $M$ iff the bivector $P$ defined as
\begin{equation}\label{poissonizat}
P \equiv \frac{1}{t} \Lambda+\frac{\partial}{\partial t} \wedge E, \;\;\;\; t\in \R_+
\end{equation}
is a Poisson structure on $M \times \mathbb{R}_+$.}
\end{thm}
Such a Poisson structure may be seen to be homogeneous, namely,  it is easy to show that $P$ in (\ref{poissonizat}) satisfies  $\mathscr{L}_Z P=-P$, with  $Z=t \frac{\partial}{\partial t}$,   the first term in (\ref{poissonizat}) being  homogeneous of degree $-1$ with respect to $t$.

On performing the change of variables $t=e^{\tau}$, the Poisson structure gets defined on $M \times \mathbb{R}$ as follows:
\begin{equation}
P=e^{-\tau}\left(\Lambda+\frac{\partial}{\partial \tau} \wedge E\right),
\end{equation}
with $\tau\in \mathbb{R}$.  This redefinition will be particularly useful for simplifying forthcoming computations. We will also consider the immersion $j: M\hookrightarrow M\times \R$ through the identification  of $M$ with $M \times \{0\}$.

The association of a Poisson structure on an extended manifold with a Jacobi structure on the original manifold  is usually referred to as Poissonization.

As it was already mentioned in the previous section, an interesting property of a contact manifold is that its Poissonization is actually a symplectic manifold, hence one could refer to it as a symplectification. Indeed, if $M$ is a contact manifold, one can define a   closed $2$-form $\omega$ on $M \times \mathbb{R}$ by using the contact form $\theta$: $\omega=d\left( e^{\tau} \pi^* \theta\right)=e^{\tau}\left(d\tau \wedge \pi^* \theta+d\pi^*\theta \right)$, where $\pi: M \times \mathbb{R} \to M$ is the projection map. 
Because of the properties of $\theta$, it is possible to prove that $\omega$ is also non-degenerate, so it is a legitimate symplectic form and makes $\left(M \times \mathbb{R} , \omega\right)$ into a symplectic manifold.

\subsection{Poisson sigma model on $M\times \mathbb{R}$}

Let us consider an $m$-dimensional Jacobi manifold $(M, \Lambda, E)$ and a Poisson sigma model having the Poisson manifold $\left(M\times \mathbb{R}, P \right)$ as target space, with Poisson structure $P=e^{-X_0}\left(\Lambda+\frac{\partial}{\partial X_0} \wedge E\right)$. The field configurations in this case are maps $X^{I}=(X^i, X^0) : \Sigma \to M \times \mathbb{R}$ and $\eta \in \Omega^1(\Sigma, X^*(T^*(M \times  \mathbb{R})))$, with $\eta_{I}=(\eta_i, \eta_0)$, where capital indices $I, J=0, \cdots m$  are related to the Poisson manifold $M \times \mathbb{R}$, 
while $i,j=1, \cdots m$  are related to the Jacobi manifold $M$.  The Poisson bi-vector field can be written explicitly in a coordinate basis as
\begin{equation}
P^{I J} = e^{-X_0}\begin{pmatrix}
& & & & -E^1\\
& & & & \\
&  & \Lambda^{ij} & &\vdots \\
& & & & \\
& & & & -E^m\\
E^1 & & \cdots &  E^m & 0
\end{pmatrix},
\end{equation}
with $P=P^{I J}\partial_{I} \wedge \partial_{J}$ and $E=E^i \partial_i$ (note that the Reeb vector field has only non-zero components  on $M$).

By splitting the equations of motion, (\ref{eqmopoisson1}) and (\ref{eqmopoisson2}) in terms of target coordinates adapted to the product manifold, one obtains:
\begin{equation}\label{eqmodecomp1}
dX^i+e^{-X^0}\left(\Lambda^{ij}\eta_j-E^i \eta_0 \right) =0,
\end{equation}
\begin{equation}\label{eqmodecomp2}
dX^0+e^{-X^0} E^i \eta_i=0,
\end{equation}
\begin{equation}\label{eqmodecomp3}
d\eta_i+\frac{1}{2} e^{-X^0} \partial_i \Lambda^{jk}\eta_j \wedge \eta_k+e^{-X^0}\partial_i E^j \eta_0 \wedge \eta_j=0,
\end{equation}
\begin{equation}\label{eqmodecomp4}
d\eta_0-\frac{1}{2} e^{-X^0} \Lambda^{jk}\eta_j \wedge \eta_k-e^{-X^0} E^j \eta_0 \wedge \eta_j=0.
\end{equation}
Let us now project the dynamics to $M$ via projection map $\pi: M \times \mathbb{R} \to M$, namely by considering $X^0= const$.
We find (by choosing $X^0=0$ for simplicity)
\begin{equation}\label{eqmopoissonization}
\begin{aligned}
{} & dX^i+\Lambda^{ij}\eta_j-E^i \eta_0=0, \\ &
E^i \eta_i=0, \\ &
d\eta_i+\frac{1}{2}\partial_i \Lambda^{jk}\eta_j \wedge \eta_k+\partial_i E^j \eta_0 \wedge \eta_j=0, \\ &
d\eta_0-\frac{1}{2}\Lambda^{jk}\eta_j \wedge \eta_k=0.
\end{aligned}
\end{equation}
where the second equation, $E^i \eta_i=0$, is purely algebraic, i.e. it is a constraint.

In  next section we will show that it is possible to derive the projected dynamics \eqn{eqmopoissonization} from  an action principle, directly defined on the Jacobi manifold, in a consistent manner. We will thus analyse the space of solutions, the algebra of constraints and the gauge invariance of the model. 

\subsection{Action principle on the Jacobi manifold}
Let $\left(M, \Lambda, E \right)$ be a Jacobi manifold, with $\Lambda \in \Gamma(\wedge^2 TM)$ and $E \in \Gamma(TM)$ satisfying  Eqs. \eqn{jacobilambda}.
We introduce the field configurations $(\phi, \eta, \lambda)$, with $\phi: \Sigma \rightarrow M$ a smooth map and $(\eta, \lambda ) \in \Omega^1(\Sigma,  \phi^*(T^*M\oplus \R))$, 
 with $T^*M\oplus \R=J^1 M$, the  vector bundle of 1-jets of real functions on $M$. Sections of the latter are isomorphic to one-forms  of the kind $e^\tau(\alpha+ f d\tau)$ \cite{Vaismanapm}, with $\alpha\in \Omega^1(M), f\in  C^\infty(M)$, $\tau$ a real parameter,  which are in turn a subalgebra of $\Omega^1(M\times \R)$.  In local coordinates   $(t,u)\in \Sigma$ we shall pose     $\phi(t,u) =X$.

 The map $\langle\;,\;\rangle$ shall indicate a pairing  between  differential forms on $\Sigma$ with values in the pull-back  $\phi^*(T^*M)$ and differential forms on $\Sigma$ with values in $\phi^*(TM)$. It is induced   by  the natural one between  $T^*M$ and $TM$ and yields  a two-form on $\Sigma$. 

\begin{prop}\label{mainprop}\emph{
The action functional 
\be\label{jacobiaction}
S(\phi, \eta, \lambda)=
 \int_\Sigma \langle\eta, \phi^*(\bd  X)\rangle + \frac{1}{2}\langle\eta,(\Lambda\circ \phi)\eta\rangle + \lambda \wedge (E\circ \phi)\eta
\ee
 with  boundary condition $\eta(u)v=0, u \in \del \Sigma, v\in T(\del \Sigma)$, defines a sigma model on the Jacobi manifold $M$, whose dynamics reproduces Eqs. \eqn{eqmopoissonization}.}
\end{prop}
Notice in particular the need for the auxiliary field $\lambda$ to take into account the contribution of the Reeb vector field. {Also, $\bd$ indicates the exterior derivative on the target manifold $M$. }
\begin{proof} 
Let us first rewrite the action as 
\be
S(X, \eta, \lambda)=
\int_{\Sigma} \left[\eta_i \wedge dX^i+\frac{1}{2}\Lambda^{ij}(X)\eta_i \wedge \eta_j-E^i(X) \eta_i \wedge \lambda \right]
\end{equation}
Prop. \ref{mainprop} is then proven by direct derivation of the equations of motion. It is a straightforward calculation to get
\begin{equation}\label{eomjacobi1}
dX^i+\Lambda^{ij}\eta_j-E^i \lambda=0,
\end{equation}
\begin{equation}\label{eomjacobi2}
d\eta_i+\frac{1}{2}\partial_i \Lambda^{jk}\eta_j \wedge \eta_k-\partial_i E^j \eta_j \wedge \lambda=0,
\end{equation}
\begin{equation}\label{eomjacobi3}
E^i \eta_i=0,
\end{equation}
which are exactly the first three equations  of  (\ref{eqmopoissonization}),  obtained from the reduction to $M$ of the Poisson sigma model on the extended manifold,  provided that we identify $\eta_0$ with $\lambda$.  

The forth apparently missing equation in  (\ref{eqmopoissonization})  is retrieved by  a consistency requirement.  On applying the exterior derivative to Eq. (\ref{eomjacobi1}),   by making use of  Eqs.  (\ref{eomjacobi2}), (\ref{eomjacobi3}), together with the explicit form  of the Schouten bracket  
\eqn{jacobiidentitygen} we find
\begin{equation}\label{eomjacobi4}
d\lambda = \frac{1}{2}\Lambda^{ij} \eta_i \wedge \eta_j.
\ee
\end{proof}

\subsubsection{Hamiltonian description, constraints and gauge transformations}

For the Hamiltonian formulation we  follow the same approach as for the Poisson sigma model.  We choose $\Sigma$ with the topology of $\mathbb{R} \times I$, with $I=[0,1]$, and pose $\beta_i=\eta_{it}$, $\zeta_i=\eta_{iu}$, $\dot{X}^i=\partial_t X^i$,  $X^{'i}=\partial_u X^i$, so that the Lagrangian of the action \eqn{jacobiaction} becomes
\begin{equation}
L(X, \zeta; \beta; \lambda)=\int_I du \left[-\dot{X}^i \zeta_i+\beta_i\left(X^{'i}+\Lambda^{ij}\zeta_j -E^i \lambda_u\right)+\lambda_t\left( E^i \zeta_i \right) \right],
\end{equation}
where $\lambda_t$ and $\lambda_u$ are the components of the  one-form $\lambda=\lambda_t dt+\lambda_u du$, with $\lambda_t,\lambda_u$ smooth functions on $\Sigma$. Explicitly, the non-trivial equations of motion read
\beqa\label{expleoms}
\dot X^i&=& -\Lambda^{ij} \beta_j+E^i \lambda_t\nonumber\\
\dot \zeta_i &=& \beta^{'}_i-\partial_i \Lambda^{jk} \beta_j \zeta_k{-\partial_i E^j \zeta_j \lambda_t+ \partial_i E^j \beta_j \lambda_u}.
\eeqa
It is evident from the Lagrangian  that   {$-\zeta_i$} is the conjugate momentum to $X^i$ with canonical  Poisson bracket{
\be\label{canb}
\{\zeta_i(u), X^j(v)\}=- {\delta_i}^j \delta(u-v)
\ee
with all  other brackets vanishing. By performing the Legendre transform with respect to the other fields, primary constraints emerge, $\pi_{\beta_i}=\pi_{\lambda_t}=\pi_{\lambda_u}=0$.   We shall  indicate with C the unconstrained phase space of maps $X^i, \beta_i, \lambda_t, \lambda_u$ and their conjugate momenta.  The Hamiltonian acquires the form 
\begin{equation}\label{hame0}
H_{\beta, \lambda}={-}\int_I du \left[\beta_i\left(X^{'i}+\Lambda^{ij}\zeta_j-E^i \lambda_u \right)+\lambda_t\left(\zeta_i E^i \right)  \right]
\ee
and conservation of constraints  imposes new ones
\beqa \label{constraints}
\mathcal{C}^i_1:&=& X^{'i}+\Lambda^{ij}\zeta_j-E^i \lambda_u=0 \label{constraint1}\\
\mathcal{C}_2:&=& \zeta_i E^i=0\label{constraint2}\\
\mathcal{C}_3:&=& \beta_i E^i=0. \label{constraint3}
\eeqa
Because of constraints $\mathcal{C}_2$ and $\mathcal{C}_3$  the last two terms in  the second equation of motion \eqn{expleoms} may be ignored, which we shall do from now on.
Analysing the whole algebra of constraints  it is possible to verify that some of them are second class, that is: $\mathcal{C}_3,  \pi_{\lambda_u}$ and the linear combinations $\beta_i \mathcal{C}^i_1$, $a^i \pi_{\beta_i}$.   Hence the number of independent first class constraints is equal to $2n$. We shall indicate the constrained phase space of maps with $\mathcal{C}={\rm C}/\sim$. }
The Hamiltonian function \eqn{hame0} is thus a combination of secondary constraints
 \begin{equation}\label{hame1}
H_{\beta, \lambda}= {-}\int_I du \, ( \beta_i \mathcal{C}_1^i + \lambda_t \mathcal{C}_2).
\end{equation}
Therefore the  model is invariant under time-diffeomorphisms. 
{It is known that gauge transformations of the Poisson sigma model do not close unless one allows for the Lagrange multipliers to depend on $X$. This is also the case for the Jacobi sigma model.  To this, we extend $\lambda=\lambda (u, X(u)), \beta=\beta(u, X(u))$. Eq. \eqn{eomjacobi4} becomes then
\be
\label{lambdabeta1}
\left(E^i \lambda_u -\Lambda^{ij} \zeta_j\right) \del_i \lambda_t - \left(E^i \lambda_t -\Lambda^{ij} \beta_j\right) \del_i \lambda_u=   \Lambda^{ij} \zeta_i   \beta_j  
\ee}
The infinitesimal generators of time diffeomorphisms are obtained through the canonical Poisson bracket on the phase space of maps,  as the Hamiltonian vector fields
{\be\label{canonicaldiff}
\xi_{H_{\beta,\lambda}}= \{  {H_{\beta,\lambda}, \cdot}\}= \dot X^i \frac{\del}{\del X^i}+ \dot\zeta_i\frac{\del}{\del\zeta_i}.
\ee}
We may state the following
\begin{prop}\label{propdiff}\emph{
The canonical Hamiltonian vector field \eqn{canonicaldiff} is the cotangent lift of the Hamiltonian vector field $\xi_{\lambda_t}\in \mathfrak{X}(PM)$ associated with $\lambda_t$ through the Jacobi bracket, according to Eq. \eqn{Hamvec}
\be
\label{hamvectfieldjacobistructure}
\xi_{\lambda_t}= (\Lambda^{ij}\del_i \lambda_t + E^j\lambda_t)\frac{\del}{\del X^j}\ee
if 
\be
\label{lambdabeta}
\beta_i = \del_i \lambda_t.
\ee}
\end{prop}
\begin{proof} 
It is easy to notice that the vector field \eqn{canonicaldiff} projects on the one defined through the Jacobi structure \eqn{hamvectfieldjacobistructure} if one uses the first of the equations of motion \eqn{expleoms} together with $\beta_i = \del_i \lambda_t$.
\end{proof}
{A sufficient condition for compatibility of Eq. \eqn{lambdabeta1} with Eq. \eqn{lambdabeta} is that $\lambda_u$ be zero, which we shall assume from now on\footnote{The general case shall be analysed in a forthcoming publication}.} 
Therefore, the Hamiltonian \eqn{hame0} acquires the form
\be
\label{hamlambda}
H_{\lambda}={-}\int_I du \left[\partial_i \lambda_t\left(X^{'i}+\Lambda^{ij}\zeta_j \right)+\lambda_t\left(\zeta_i E^i \right)  \right].
\end{equation}
Moreover,  the Hamiltonian vector fields associated with the Jacobi structure close under Lie bracket,  it being  (see Sec. \ref{Jacobiman})
\be
[\xi_{\lambda_t},  \xi_{\gamma_t}]= \xi_{\{\lambda_t,\gamma_t\}_J}.
\ee
Thus, we can state the following
\begin{thm}\label{prop2}\emph{
\mbox{}
\begin{enumerate}[(i)]
\item 
For each $\lambda$ 
there exists an Hamiltonian vector field  associated with $H_\lambda$ such that its projection onto $PM$  is the Hamiltonian vector field $\xi_{\lambda_t}$ associated with $\lambda_t$ through Jacobi bracket. The vector fields $\xi_{\lambda_t}$ are  infinitesimal generators of space-time diffeomorphisms.
\item The map
$\;
H_\lambda\rightarrow \xi_{\lambda}\;
$
is a Lie-algebra homomorphism
\be
\{H_\lambda, H_{\tilde \lambda}\}= H_{\{\lambda_t, \tilde \lambda_t\}_J}
\ee
where $\{\;,\;\}$ is the canonical Poisson bracket on the phase space of maps, $\{\;,\;\}_J$ is the Jacobi bracket on $PM$ (with $PM\ni X$ the configuration space).
\item 
The Jacobi sigma model  is gauge-invariant under space-time diffeomorphisms.
\end{enumerate}}
\end{thm}


\begin{proof} 

The first statement is a direct consequence of Prop. \ref{propdiff}. 
The Hamiltonian is a pure constraint, therefore the model is invariant under time diffeomorphisms generated by  $\xi_{\lambda_t}$. Similarly to the Poisson sigma model, space-diffeomorphisms are recovered if one considers that the action  is invariant under  $\dot X\leftrightarrow X'$ , $\beta \leftrightarrow  - \zeta$ and  $\lambda_t\leftrightarrow \lambda_u$.  
\\
The second  statement is proven by  direct calculation. 
We first compute
\beqa
\{\mathcal{C}^i_1(u), \mathcal{C}_2(v) \}&=&\del_u \delta(u-v) E^i + \Lambda^{k j} \del_k E^i \zeta_j \delta(u-v)
\\
\{\mathcal{C}^i_1(u), \mathcal{C}^j_2(v) \}&=&\left(\Lambda^{il}\del_l \Lambda^{jm} -\Lambda^{jl}\del_l \Lambda^{im} \right)\zeta_m \delta(u-v)
\eeqa
We have then 
\beqa
\{H_\lambda, H_{\tilde \lambda}\}&=& \int du\,du'\, \mathscr{L}_{\xi_{H_\lambda}} H_{\tilde \lambda}
 \nonumber \\
&=& \int du \, du' \left[ \mathcal{C}_1^i \del_i \{\lambda_t, \tilde \lambda_t\}_J+ \mathcal{C}_2 \{\lambda_t, \tilde \lambda_t\}_J \right] = H_{\{\lambda_t, \tilde \lambda_t\}_J}.
 \eeqa
 The latter implies that the Hamiltonian constraints are first class, thus generating gauge transformations with infinitesimal generators $\xi_{\lambda_t}$ and Lie bracket
\be
[\xi_{\lambda_t},\xi_{\tilde\lambda_t}]= \xi_{\{{\lambda_t},\tilde\lambda_t\}_J}.
\ee
The last statement is therefore proven.  
 \end{proof}

Now we are in a position to prove the remarkable  result that the reduced phase space $\mathcal{G}$ is finite dimensional. In fact, we have the following 
\begin{thm}\label{mainthm}\emph{
Let $(X, \zeta) \in \mathcal{C} \subset T^* PM$. The subspace of $T_{(X, \zeta)}(T^*PM)$ spanned by the Hamiltonian vector fields $\xi_{\beta, \lambda_t}$   is a closed subspace of codimension $2 \text{dim}(M) {-2}$.}
\end{thm}
\begin{proof}
We follow for the proof the same approach as in  \cite{Cattaneo2001} where the theorem is shown to hold for the Poisson sigma model \eqn{sigmaP}.

Let us consider the subspace $\mathcal{S}_{(X, \zeta)}$ of $T_{(X,\zeta)}\mathcal{C}$ spanned by the Hamiltonian vector fields $\xi_{\beta, \lambda_t}$.
The map 
$(\beta, \lambda_t) \rightarrow \xi_{ \beta, \lambda_t}$, explicitly given by 
\beqa
\delta_{\xi_{H} }X^i & \coloneqq& \{H_{\lambda_t}, X^i \}=-\Lambda^{ij}  \beta_j+E^i \lambda_t \label{trasf1}\\
\delta_{\xi_{H} } \zeta_i &\coloneqq& \{H_{\lambda_t}, \zeta_i \}=(\beta_i)'-\partial_i \Lambda^{jk} \beta_j \zeta_k+\partial_i E^j \zeta_j \lambda_t \label{trasf2}
\eeqa
 is linear. However, on the constraint manifold $\mathcal{C}$, the last term in the r.h.s. of \eqn{trasf2} vanishes: \begin{equation*}
\partial_i E^j \zeta_j \lambda_t=-E^j \partial_i \zeta_j \lambda_t=0,
\end{equation*}
where the first equality comes from the constraint $\zeta_i E^i=0$, while the second one follows from the fact that $X$ and $\zeta$ are canonically conjugated. Therefore, the components of the map are given by
\begin{equation}\label{hamvectcomp1}
\xi_1^i \coloneqq \xi_{\beta, \lambda_t} X^i =-\Lambda^{ij}  \beta_j+E^i \lambda_t ,
\end{equation}
\begin{equation}  \label{hamvectcomp2}
\xi_{2,i} \coloneqq \xi_{ \beta, \lambda_t} \zeta_i=(\beta_i)'-\partial_i \Lambda^{jk}\beta_j \, \zeta_k.
\end{equation}
Let us start by analysing the kernel of this linear map. In particular, from $\xi_{\beta, \lambda_t}=0$ we obtain
\begin{equation}
-\Lambda^{ij} \beta_j+E^i \lambda_t=0
\end{equation}
and 
\begin{equation}
(\beta_i)'-\partial_i \Lambda^{jk} \beta_j \, \zeta_k=0.
\end{equation}
The second one is a homogeneous linear first order ODE with initial condition $( \beta(0)=0$ (because of the boundary conditions), so the solution vanishes identically. The first one is an algebraic relation for which, by using the previous result, we have $E^i \lambda_t=0$ and since the Reeb vector field is nowhere vanishing we have $\lambda_t=0$. Hence, the map is injective and it is sufficient to look at the image space. 

The tangent vector $(\tilde{X}, \tilde{\zeta})$ to a point $(X, \zeta) \in \mathcal{C}$ is the solution of the linearized constraints
\begin{equation}\label{linearizedconstr1}
\tilde{X}'^i+{A_j}^i \tilde{X}^j+\Lambda^{ij}\tilde{\zeta}_j=0,
\end{equation}
\begin{equation}\label{linearizedconstr2}
E^i \tilde{\zeta}_i+\partial_i E^j \zeta_j \tilde{X}^i\underset{\mathcal{C}}{=}E^i \tilde{\zeta}_i=0,
\end{equation}
where we defined ${A_i}^j=\partial_i \Lambda^{jk}\zeta_k$. If $(\tilde{X}, \tilde{\zeta})$ is Hamiltonian (i.e. it is in the image of $\xi$), then we have
\begin{equation}
\tilde{X}^i=-\Lambda^{ij} \beta_j+E^i \lambda_t,
\end{equation} 
\begin{equation}\label{proofeq1}
\tilde{\zeta}_i=(\beta_i)'-{A_i}^j \beta_j,
\end{equation}
and in particular
\begin{equation*}
\tilde{X}^i(0)-E^i (X(0))\lambda_t(0)=0.
\end{equation*} 
If we introduce the matrix $V=\hat{P}\exp[-\int A\, du ]$ as the path-ordered exponential of $A$, i.e. the solution of the differential equation
\begin{equation}\label{defvmatrix}
\begin{cases}
(V_i^j)'=-V_i^k(u) {A_k}^j(u) \\
V_i^j(0)=\delta_i^j,
\end{cases}
\end{equation}
then Eq. (\ref{proofeq1}) can be written as
\begin{equation}
\tilde{\zeta}_i (u)=(V^{-1}(u))_i^k \, \partial_u[V(u) ^j_k \beta_j (u)].
\end{equation}
From this equation we can define the $m$ functions
\begin{equation}
p_i(u) \coloneqq \int_0^u dv V(v)^k_i \tilde{\zeta}_k(v)= \int_0^u \partial_v[V(v) ^k_i  \beta_k (v)],
\end{equation}
from which it follows that 
\begin{equation*}
\int_I du \, V(u)_i^k \tilde{\zeta_k}(u)=0.
\end{equation*}
Hence, we conclude that if $(\tilde{X}, \tilde{\zeta})$ is in the image of $\xi$, then we have 
\begin{equation}\label{conditionsproofdim}
\tilde{X}^i(0)-E^i (X(0))\lambda_t(0)=0, \quad \ \quad  \int_I du \,  V(u)_i^k \tilde{\zeta_k}(u)=0.
\end{equation} 
If we choose a basis of vector fields with $E$ as one of the basis vectors, say the $m-th$, it is possible to see that these are $2m-2$ independent conditions.  {Indeed, by posing $a= 1, ..., m-1$, the first of Eqs. \eqn{conditionsproofdim} yields
\be 
 \tilde{X}^a(0)=0, \;\;\;  \tilde{X}^m(0) =\lambda_t(0)
\ee
but the latter is not a gauge invariant statement. As for the second of Eqs. \eqn{conditionsproofdim} we need to go back to the constraints $\mathcal{C}_2, \mathcal{C}_3$ which now imply $\beta_m=\zeta_m=0$. Then the whole derivation above, starting from Eq. \eqn{defvmatrix} has to be repeated for the sole indices ranging from 1 to $m-1$. We end up with $m-1$ functions $p_a(u)$, yielding $m-1$ invariants.}

Viceversa, given $(\tilde{X}, \tilde{\zeta})$ a tangent vector at the point $(X, \zeta) \in \mathcal{C}$ satisfying the conditions in Eq. (\ref{conditionsproofdim}), then $(\tilde{X}, \tilde{\zeta})$ is an Hamiltonian vector field with the choice $ \beta_i=(V^{-1})^k_i \, p_k=(V^{-1})^k_i \int_0^u dv \, V(v)^{\ell}_k \tilde{\zeta}_{\ell}(v)$. In fact, let us define a vector field 
\begin{equation}\label{defYfield}
Y^i(u)=-\Lambda^{ij}(u)\beta_j (u) + E^i(u) \lambda_t(u),
\end{equation}
(hence fulfilling  $Y^i(0)=E(0)^i \lambda_t (0)$), with 
\begin{equation}\label{defbetaproof}
 \beta_i=(V^{-1})^k_i \int_0^u dv \, V(v)^{\ell}_k \tilde{\zeta}_{\ell}(v).
\end{equation}
Let us check that $Y$ satisfies the same ODE as $\tilde X$ with the same boundary condition, hence it is the same field. To this we compute
\begin{equation*}
\begin{aligned}
Y'^i {} & =-\partial_k \Lambda^{ij} X'^k (V^{-1})^p_j \int_0^u dv V_p^{\ell} \tilde{\zeta}_{\ell}-\Lambda^{ij}\left[\partial_u (V^{-1})^p_j \int_0^u dv V_p^k \tilde{\zeta}_k+(V^{-1})^p_j V_p^k \tilde{\zeta}_k \right] \\ & +\partial_k E^i X'^k \lambda_t+E^i \lambda'_t,
\end{aligned}
\end{equation*}
and by using the relation $\lambda'_t=-\Lambda^{ij}\beta_i \zeta_j$ and the constraint equation $X'^i=-\Lambda^{ij} \zeta_j$, we have
\begin{equation*}
Y'^i = \beta_j \zeta_p \left( \Lambda^{kp} \partial_k \Lambda^{ij}+\Lambda^{ik}\partial_k \Lambda^{pj}-E^i \Lambda^{pj}\right)-\Lambda^{ij} \tilde{\zeta}_j-\partial_k E^i \Lambda^{kp} \zeta_p \lambda_t,
\end{equation*}
where we also used the choice of $ \beta$ in Eq. (\ref{defbetaproof}) and the relation $\partial_u {(V^{-1})^p}_j=-(V^{-1})^k_j \partial_u V_k^{\ell} (V^{-1})^p_{\ell}$, as well as the defining equation of $V$ (\ref{defvmatrix}). At this point we can use Eq. (\ref{defYfield}) and the relation for $[\Lambda, \Lambda]_S$ in Eq. (\ref{jacobiidentitygen}) to obtain
\begin{equation*}
Y'^i=-\partial_k \Lambda^{ip} \zeta_p Y^k-\Lambda^{ij}\tilde{\zeta}_j+\left(\partial_k \Lambda^{ip} E^k-\partial_k E^i \Lambda^{kp} \right) \zeta_p \lambda_t.
\end{equation*}
Using the fact that $\mathscr{L}_E \Lambda=0$, namely Eq. (\ref{liederivelambda}), we have
\begin{equation*}
\left(\partial_k \Lambda^{ip} E^k-\partial_k E^i \Lambda^{kp} \right) \zeta_p \lambda_t=\Lambda^{ik} \partial_k E^p \zeta_p \lambda_t,
\end{equation*}
but   $\partial_k E^p  \zeta_p=0$ on $\mathcal{C}$,  so that  finally $Y$ satisfies the linearised constraint in Eq. (\ref{linearizedconstr1}) with the same boundary condition. One can also easily prove that the linearised constraint in Eq. (\ref{linearizedconstr2}) is satisfied as well. In fact, contracting Eq. (\ref{hamvectcomp2}) with the Reeb vector field we obtain
\begin{equation*}
\begin{aligned}
{} & E^i (\beta_i)'-E^i \partial_i \Lambda^{jk} \beta_j\zeta_k \\ &= E^i (\beta_i)'- \beta_j \zeta_k \Lambda^{pk} \partial_p E^j=E^i ( \beta_i)'+ \beta_i\partial_j E^i X'^j=(E \cdot \beta)'=0,
\end{aligned}
\end{equation*}
where in the first equality we used the expression for the Lie derivative in Eq. (\ref{liederivelambda}) and again $\partial_k E^p  \zeta_p \underset{\mathcal{C}}{=}0$, while the last equality follows from the constraint (\ref{eomjacobi3}).

To conclude, we have proved that the image of $\xi$ is the subspace spanned by $\xi_{\beta, \lambda_t}$ modulo the $2m {-2}$ conditions in Eq. (\ref{conditionsproofdim}), i.e. it is a closed subspace of codimension $2m {-2}$.
\end{proof}

To summarise the results, the reduced phase space of the model is $\mathcal{G}= \mathcal{C}/\text{Diff}(\Sigma)$, where  $\mathcal{C}$ indicates the space of solutions of Eqs. \eqn{constraints} while $\text{Diff}(\Sigma)$ is the gauge group of diffeomorphisms generated by the constraints, which are in turn associated with the Jacobi structure. The dimension of $\mathcal{G}$ is finite and equal to $2 \text{dim}(M) {-2}$.


\section{Metric extension and Polyakov action}\label{secfulljacobi}

We will show in this section that, just like in the Poisson sigma model case \cite{Schupp2012}, the topological model considered so far can be  generalised into a non-topological model by introducing a dynamical term containing the metric of the worldsheet (via the Hodge star operator on $\Sigma$) and a metric tensor $G$ for the target space:
\begin{equation}
S(X, \eta, \lambda)=\int_{\Sigma} \left[\eta_i \wedge dX^i+\frac{1}{2}\Lambda^{ij}(X) \,\eta_i \wedge \eta_j-E^i(X) \,\eta_i \wedge \lambda +\frac{1}{2}(G^{-1})^{ij}(X) \, \eta_i \wedge \star \eta_j \right].
\end{equation}
We are now concerned with the integration of the auxiliary fields ($\eta$ and $\lambda$) to obtain a Polyakov action for the embedding maps $X$. To do this, we first write the new equations of motion following by the introduction of the new metric term:
\begin{equation}\label{eometric1}
dX^i+\Lambda^{ij}\eta_j-E^i \lambda+(G^{-1})^{ij} \star \eta_j=0,
\end{equation}
\begin{equation}\label{eometric2}
d\eta_i+\frac{1}{2}\partial_i \Lambda^{jk}\eta_j \wedge \eta_k-\partial_i E^j \eta_j \wedge \lambda+\frac{1}{2}\partial_i (G^{-1})^{jk}\eta_j \wedge \star \eta_k=0,
\end{equation}
\begin{equation}\label{eometric3}
E^i \eta_i=0.
\end{equation}
Thanks to the new term and the fact that the metric tensor $G$ is naturally non-degenerate, the equation for $\eta$ can be extracted from Eq. (\ref{eometric1}):
\begin{equation}\label{eqstareta}
\star \eta_j=-G_{ij}\left(dX^i+\Lambda^{ik}\eta_k-E^i \lambda \right).
\end{equation}
On  applying again the Hodge star operator (we choose the metric signature $(1,-1)$ for $\Sigma$, so in this case $\star^2=\mathds{1}$) and substituting back the expression (\ref{eqstareta}) for $\star \eta$ we have 
\begin{equation}\label{eqeta}
\eta_p=-{(M^{-1})^j}_p G_{ij}\left(\star dX^i-\Lambda^{ik}G_{\ell k} dX^{\ell}+\Lambda^{ik}G_{\ell k}E^{\ell}\lambda-E^i \star \lambda \right),
\end{equation}
where we defined the matrix ${M^p}_j={\delta^p}_j-G_{j i}\Lambda^{ik}G_{k \ell}\Lambda^{\ell p}$, which is symmetric and assumed to be non-degenerate, without any assumption on the non-degeneracy of the $\Lambda$ bivector. 

Remarkably,  by substituting  the expression for $\star \eta$ into the term $\frac{1}{2}(G^{-1})^{ij} \eta_i \wedge \star \eta_j$, the action acquires the simple form
\begin{equation}
S=\frac{1}{2}\int_{\Sigma} \eta_i \wedge dX^i,
\end{equation}
where we also used the fact that on-shell $E^i \eta_i=0$. Replacing the explicit expression  for $\eta$,  Eq. (\ref{eqeta}), in the action we obtain
\begin{equation}\label{actionlambda}
\begin{aligned}
S(X, \lambda) = \int_{\Sigma} {} & \bigg[\frac{1}{2}{(M^{-1})^p}_i G_{jp} \, dX^i \wedge \star dX^j- \frac{1}{2}{(M^{-1})^p}_i G_{\ell p} \Lambda^{\ell k} G_{jk} \, dX^i \wedge dX^j \\ & - \frac{1}{2}{(M^{-1})^p}_i G_{\ell p}\Lambda^{\ell k} G_{mk} E^m \lambda \wedge dX^i+ \frac{1}{2}{(M^{-1})^p}_i G_{\ell p} E^{\ell} \star \lambda \wedge dX^i \bigg].
\end{aligned}
\end{equation}
There is still $\lambda$ to be integrated out. This can be achieved by recognising  $\int_{\Sigma} \star \lambda \wedge dX$ as the  scalar product on the space of $1-$forms so that $\int_{\Sigma} \star \lambda \wedge dX=-\int \lambda \wedge \star dX$. Thus  the last two terms in Eq. (\ref{actionlambda}) are proportional to $\lambda$. The latter acting as a Lagrange multiplier, imposes the constraint 
\begin{equation}\label{constraintpolyakov}
(M^{-1})_{i \ell}\left( \Lambda^{\ell k}G_{mk}E^m dX^i+E^{\ell} \star dX^i\right)=0
\end{equation}
where we used the metric tensor $G$ to lower and  raise the target space indices.
This means that on-shell the term proportional to  $\lambda$ vanishes and what remains is the second order action
\begin{equation}\label{PolJac}
S=\int_{\Sigma} \left[g_{ij} dX^i \wedge \star dX^j+B_{ij} dX^i \wedge dX^j \right]
\end{equation}
with metric and $B$-field given by:
\begin{equation}
g_{ij}=G_{jp} {(M^{-1})^p}_i, \quad B_{ij}=G_{ik}{(M^{-1})^p}_j G_{ p \ell}\Lambda^{\ell k}.
\end{equation}
Eq. \eqn{PolJac} represents a Polyakov string action with target space a Jacobi manifold, with the Jacobi structures hidden in the metric and B-field. Note that the Reeb vector field $E$ plays no role in the definition of $g$ and $B$ but it is present in the constraint (\ref{constraintpolyakov}).

\section{Jacobi sigma model on SU(2)}\label{secExamples}

In this section we consider the group manifold of $SU(2)$ as target space. It provides an example of a contact manifold where the contact structure can be taken to be one of the left-invariant basis one forms of the group, say $\theta^i$ so that $\ell^{-1}\bd \ell=\theta^i e_i \in \Omega^1(SU(2), \mathfrak{su}(2))$, is the Maurer-Cartan left-invariant  one form on the group,  with $\ell\in SU(2)$, $e_i $ the Lie algebra generators and choose, to be definite, $\vartheta=\theta^3$ as contact structure for $SU(2)$ (right invariant one-forms could be used equivalently).
It is easily checked that it  satisfies the conditions \eqn{contactcond}-\eqn{contactcond3}. Indeed, the Maurer-Cartan equation $\bd\theta^k=\frac{1}{2} {\epsilon^k}_{ij} \theta^i \wedge \theta^j$ leads to 
\begin{equation}
\bd\vartheta=\theta^1 \wedge \theta^2,
\end{equation}
so that $\vartheta \wedge \bd\vartheta=\theta^1 \wedge \theta^2 \wedge \theta^3 = \text{Vol}_{S^3}$.

From Eqs. (\ref{contactcond}) we get for the Jacobi structure
\begin{equation}
\Lambda=Y_1 \wedge Y_2, \quad E=Y_3,
\end{equation}
where $Y_i$, with $i \in \{1,2,3 \}$ are the left-invariant vector fields on $SU(2)$.

In order to define  the fields $X^i, \dot X^i, X'^i$ in a chart independent way, we resort to the group valued map $\phi:\Sigma \rightarrow SU(2)$,  and the pull-back map  $\phi^*: \Omega^1(SU(2))\rightarrow \Omega^1(\Sigma)$,  so to get
\be
\phi^*(g^{-1} \bd g)= (g^{-1} \del_t g) dt + (g^{-1} \del_u g) du
\ee
which is a one-form on $\Sigma$, valued in the Lie algebra of $SU(2)$. We shall omit the pull-back from  now on, but it will be always understood, unless otherwise stated.
The action  \eqn{jacobiaction} may thus be written as follows 
 \be
 S(g, \eta, \lambda)=\int_{\Sigma} \left[\eta_i \wedge (g^{-1}dg)^i +\frac{1}{2}\epsilon^{3ij}\eta_i \wedge \eta_j-\eta_3 \wedge \lambda \right]
\end{equation}
where $g^{-1} dg= (g^{-1} dg)^ie_i$ is Lie algebra valued, while $\eta= \eta_i e^{i*}$ is valued in the dual of the Lie algebra and $e^{i*}(e_j)= \delta^i_j$.

The equations of motion acquire the form
\begin{equation}
(g^{-1}dg)^i+\epsilon^{3ij}\eta_j-\delta^{i3}\lambda=0
\end{equation}
and 
\begin{equation}
{d\eta_i+ \eta_i \wedge \eta_3 + \epsilon_{3i}^j \eta_j\wedge\lambda=0, \quad\quad \eta_3=0}.
\end{equation}

Differently from Poisson sigma models, despite the fact that $\Lambda$ is degenerate, the field  $\eta$ can be integrated out from the action according to the following procedure, which is valid for any contact manifold.  By using Eqs. (\ref{contactcond2})-(\ref{contactcond3}) we can contract the equation of motion in Eq. (\ref{eomjacobi1}) with $\vartheta= \vartheta_i \theta^i$ to obtain
\begin{equation}
\langle \vartheta, g^{-1}dg \rangle-\lambda \langle \vartheta, E \rangle = (\vartheta)_i (g^{-1}dg)^i-\lambda=0,
\end{equation}
with $\langle \cdot , \cdot \rangle$ the natural pairing between $T^*M$ and $TM$, 
so that $\lambda=\vartheta_i (g^{-1}dg)^i$ can be integrated out. To integrate the fields $\eta$ we can contract again Eq. (\ref{eomjacobi1}) with $d\vartheta$, and using again Eqs. (\ref{contactcond2})-(\ref{contactcond3}) we obtain
\begin{equation}
-(d\vartheta)_{ij} (g^{-1}dg)^j+\eta_i=0,
\end{equation}
so that $\eta_i=(d\vartheta)_{ij}(g^{-1}dg)^j$ can be integrated out as well. Substituting the expressions for $\lambda$ and $\eta$ in the action in Eq. (\ref{jacobiaction}), we obtain the second order action
\begin{equation}
S_2=-\frac{1}{2}\int_{\Sigma}  \langle d\vartheta ,(g^{-1}dg) \wedge (g^{-1}dg) \rangle,
\end{equation}
which has the same form as an A-model. In particular, $d\vartheta$ has the role of a $B$-field, and in this case it is closed. In particular, by writing $d\vartheta$ explicitly and by further introducing the notation
\be
(g^{-1} \del_t g)= A^i e_i, \;\;\;\; (g^{-1} \del_u g)= J^i e_i
\ee 
with $(A^i, J^i)$ the  currents of the sigma model, we have
\begin{equation}
S_2=-\frac{1}{2}\int_{\Sigma} \epsilon_{3ij} (g^{-1}dg)^i \wedge (g^{-1}dg)^j=\int_{\Sigma}  d^2u \, \epsilon_{3ij} A^i J^j.
\end{equation}

It is also interesting to specialise the Polyakov action obtained in Eq. (\ref{PolJac}) to the $SU(2)$ target manifold by introducing the natural Cartan-Killing metric on the latter: $G_{ij}=\delta_{ij}$. By using $G_{ij}=\delta_{ij}$ and $\Lambda^{ij}=\epsilon^{3ij}$, the metric $h$ and $B$-field are then obtained as 
\begin{equation}\label{backg}
h_{ij}= \delta_{ij}-\frac{1}{2}\epsilon_{ik3}\delta^{kl}\epsilon_{jl3}, \quad B_{ij}=-\frac{1}{2}\epsilon_{3ij},
\end{equation}
so to have 
\begin{equation}
S=\int_{\Sigma} \left[h_{ij} (g^{-1}dg)^i \wedge \star (g^{-1}dg)^j-\frac{1}{2}\epsilon_{3ij} (g^{-1}dg)^i \wedge (g^{-1}dg)^j \right],
\end{equation}
complemented with the constraint
\be
(g^{-1}dg)^3=0,
\ee
which follows from Eq. (\ref{constraintpolyakov}).

It is interesting to note the metric $h$ in Eq. \eqn{backg}. This has been already  obtained in the context of  Poisson-Lie duality  of $SU(2)$ sigma models \cite{Marotta2019, Bascone2020, PV19, Bascone2020a,   Marotta2018, Pezzella2019} as a non-degenerate metric for   dual sigma models  with target manifold the group $SB(2, \C)$.  The latter  plays a role in the  Drinfel'd double decomposition of the group $SL(2,\mathbb{C})$.  Therefore it is an interesting question, which we leave for further investigation, to understand what is  the relation, if any, between the two models: the dynamical Jacobi sigma model on the manifold of the group $SU(2)$ and the $SU(2)$ sigma model with Poisson-Lie duality made explicit. 

\section{Conclusions and Outlook}

We have defined and analysed  a two-dimensional sigma model with target space a Jacobi manifold, as a natural generalisation of a Poisson sigma model. In particular, we started from the concept of Poissonization of a Jacobi manifold, which consists in the construction of a homogeneous Poisson structure on the extended manifold $M \times \mathbb{R}$ from a Jacobi structure on $M$. We projected the dynamics of this extended Poisson sigma model on the Jacobi manifold $M$ and then formulated a new sigma model action having $M$ as target space which reproduces the projected dynamics. This is schematically illustrated in the diagram \ref{fig:diagram}.
\begin{center}
\begin{figure}[ht]
\centering
\includegraphics[width=0.5\linewidth]{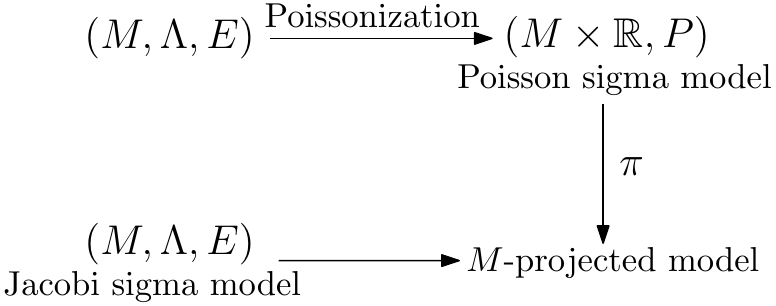}
\caption{}
\label{fig:diagram}
\end{figure}
\end{center}
We have analysed the Hamiltonian formulation of the model, which exhibits first class constraints generating gauge transformations. In particular, we have shown that Hamiltonian vector fields associated with the  Jacobi structure can be associated with gauge transformations generating space-time diffeomorphisms, and the model is topological. The reduced phase space of the model, which is the constrained manifold modulo gauge symmetries, has finite dimension equal to $2\rm{dim}M -2$.

We also investigated the possibility of including  a metric term in the action, resulting in a non-topological sigma model.  The auxiliary fields can be integrated out to give a Polyakov action, where the metric and $B$-field are related to the defining structures of the target Jacobi manifold. 

In particular, we analysed $SU(2)$ as an example of contact target manifold in view of its relation with Poisson-Lie symmetry and T-duality. 

Issues such as quantisation, integrability and T-duality of the Jacobi model represent interesting directions of research, some of which are presently under investigation.

 \vspace{5pt}

\noindent{\bf Acknowledgements} 
We are deeply indebted to Alberto Ibort and Giuseppe Marmo for having introduced us to Jacobi manifolds and having encouraged  to investigate their application to sigma models. 

\end{document}